\documentclass[11pt]{article}

\usepackage{a4wide}
\usepackage{amsmath,amssymb,amsthm}
\usepackage{enumerate}

\newcommand{\setN}{\mathbb{N}}
\newcommand{\pre}[1]{{}^\bullet{#1}}
\newcommand{\post}[1]{{#1}^\bullet}
\newcommand{\restr}[1]{|_{#1}}
\newcommand{\gt}[1]{\xrightarrow{#1}}
\newcommand{\rset}[1]{[#1\rangle}
\newcommand{\prior}{_{\textsc{pr}}}

\theoremstyle{plain}
\newtheorem*{theorem*}{Theorem}

\title{A concise proof of Commoner's theorem}

\author{Petr Jan\v{c}ar\\
{\small Dept of Computer Science, Faculty of Science, Palack\'y University in
Olomouc, Czechia}\\
\small{petr.jancar@upol.cz}}

\begin{document}

\date{}
\maketitle

\begin{abstract}
\noindent
The textbook proofs of Commoner's theorem characterizing liveness in
	free-choice Petri nets
are given in contexts of technical notions and claims that make the
	proofs look a~bit long. The aim of this note is to give a concise
	self-contained proof.
\end{abstract}	

This note is a slight elaboration of the (non-reviewed) $1$-page text
``Petr Jan\v{c}ar, A concise proof of {C}ommoner's theorem'' in 
Petri Net Newsletter, No 49, page 43 (October 1995), 
which referred to~\cite{DeselEsparza1995,Reisig1985}.
The aim here is to make that text easily accessible, by putting a~more
readable elaboration on arXiv (with no ambition to fit everything on
one page).

We first recall basic notions and notation, also making some easy
observations that are useful in the proof of the theorem that follows.

A \emph{net} $N$ is a triple $(P,T,F)$ where $P$ and $T$ are finite
disjoint sets of \emph{places} and \emph{transitions}, respectively,
and $F\subseteq(P\times T)\cup(T\times P)$ is a \emph{flow relation}.
For $x\in P\cup T$ and $X\subseteq P\cup T$ we put $\pre{x}=\{y\mid
(y,x)\in F\}$, $\pre{X}=\bigcup_{x\in X}\pre{x}$, and 
 $\post{x}=\{y\mid
(x,y)\in F\}$, $\post{X}=\bigcup_{x\in X}\post{x}$. 
A set $Q\subseteq P$ is a \emph{trap} if $\post{Q}\subseteq\pre{Q}$;
a set $S\subseteq P$ is a \emph{siphon} if  $\pre{S}\subseteq\post{S}$.

A~\emph{marking} $M$ of $N$ is a finite multiset of places, hence
$M\in\setN^P$; by $\mathbf{0}$ we denote the empty multiset.
For $R\subseteq P$, $M\restr{R}$ denotes the restriction of $M$ to
$R$.
For $t\in T$ we have $M\gt{t}M'$ if $t$ is \emph{enabled at} $M$, 
i.e.\ $\pre{t}\subseteq
M$, and $M'=M-\pre{t}+\post{t}$. For $\sigma\in T^*$ we define 
$M\gt{\sigma}M'$ inductively: $M\gt{\varepsilon}M$ (where $\varepsilon$
is the empty word), and if $M\gt{t}M'$
and
$M'\gt{\sigma}M''$ then $M\gt{t\sigma}M''$; we say that $\sigma$ is
\emph{enabled at} $M$ if $M\gt{\sigma}M'$ for some $M'$. We put
$\rset{M}=\{M'\mid M\gt{\sigma}M'$ for some $\sigma\in T^*\}$.
Let us observe that if $Q$ is a~trap and $M\restr{Q}\neq\mathbf{0}$, then 
$M'\restr{Q}\neq\mathbf{0}$ for all $M'\in\rset{M}$;
if $S\subseteq P$ is a siphon and $M\restr{S}=\mathbf{0}$ then 
 $M'\restr{S}=\mathbf{0}$ for all $M'\in\rset{M}$.

A~\emph{transition} $t$ is \emph{dead at} $M$ if $\forall
M'\in\rset{M}:\pre{t}\not\subseteq M'$ ($t$ is not enabled at
$M'$); by $D_M$ we denote the set of transitions that are dead at $M$.
A~\emph{transition} $t$ is \emph{live at} $M$ if $\forall
M'\in\rset{M}:t\not\in D_{M'}$ (hence $\forall
M'\in\rset{M}\, \exists M''\in\rset{M'}:t$ is enabled at $M''$); by $L_M$ we denote the set of
transitions that are live at $M$. We observe that $M'\in\rset{M}$
entails $D_M\subseteq D_{M'}$ and $L_M\subseteq L_{M'}$; moreover,
for any $M$ there is
$M'\in\rset{M}$ such that $D_{M'}\cup L_{M'}=T$. 

A net $N=(P,T,F)$
is a~\emph{free-choice net} if 
$\pre{t}\cap\pre{t'}\neq\emptyset$ entails 
$\pre{t}=\pre{t'}$ (for all $t,t'\in T$).
Hence the free-choice property guarantees that at each $M$ 
all input-sharing transitions $t,t'$ are 
either both enabled or both disabled.
Here we also observe that 
if $t\in D_M$, then
$\post{(\pre{t})}\subseteq D_M$; 
moreover, 
there is $p\in\pre{t}$ such that
 $M(p)=0$ and $p\not\in \post{(L_M)}$.

\begin{theorem*}[Commoner]
For any free-choice net $N=(P,T,F)$ with no isolated places
	(i.e., $P=\pre{T}\cup\post{T}$)
	and any $M_0\in\setN^P$, the next
two conditions are equivalent:
	\begin{enumerate}[a)]
		\item
			all transitions are live at $M_0$
			(i.e., $L_{M_0}=T$);
		\item
		for every nonempty siphon $S\subseteq P$ 
there is a trap $Q\subseteq S$ such that
			$M_0\restr{Q}\neq \mathbf{0}$.
	\end{enumerate}
	\end{theorem*}
\begin{proof}
We fix some assumed $N=(P,T,F)$ and $M_0$, and show the following
 implications.
\medskip

	1. non-a)$\Rightarrow$non-b).
\\
We assume that $L_{M_0}\neq T$; we then choose $M\in\rset{M_0}$ such that
	$D_M\neq\emptyset$ and $D_M\cup L_M=T$. For every $t\in D_M$
	we fix one $p_t\in\pre{t}$ such that $M(p_t)=0$ and $p_t\not\in
	\post{(L_M)}$
	(which is possible by the respective above observation).
Hence $S=\{p_t\mid t\in D_M\}$ is a~nonempty
	siphon ($\pre{S}\subseteq D_M\subseteq\post{S}$)
	such that $M\restr{S}=\mathbf{0}$; hence $S$ cannot
	contain a trap $Q$ such that $M_0\restr{Q}\neq\mathbf{0}$ 
	(otherwise we would have $M\restr{Q}\neq\mathbf{0}$ since
	$M\in\rset{M_0}$).

\medskip

	2. non-b)$\Rightarrow$non-a).

	Let $S$ be a nonempty siphon (hence $S\neq\emptyset$ and
	$\pre{S}\subseteq\post{S}$) and $Q\subseteq S$ be the maximal trap,
	i.e.\ the union
	of all traps, inside $S$ (hence $\post{Q}\subseteq\pre{Q}$,
	and we can have $Q=\emptyset$);
	let $M_0\restr{Q}=\mathbf{0}$.
Since there are no isolated places, we must have
	$\post{S}\neq\emptyset$.
 We will finish the proof by showing that
	$\post{S}\cap L_{M_0}=\emptyset$.
	This is obvious if $Q=S$, in which case $\post{S}\subseteq
	D_{M_0}$, so we assume $Q\subsetneq S$. 

First we observe that for any $R$ such that $Q\subsetneq R\subseteq
	S$, which entails that $R$ is not a~trap,
	there is $t$ such that $\pre{t}\cap
	R\neq\emptyset$ and
			$\post{t}\subseteq P\smallsetminus R$
	(i.e., $t\in\post{R}\smallsetminus\pre{R}$),
	which also entails that  $\pre{t}\cap Q=\emptyset$
	($t\not\in\pre{Q}$ entails
	$t\not\in\post{Q}$).
For some $k\leq |S\smallsetminus Q|$ we can thus fix a sequence
	$t_1,t_2,\cdots,t_k$ 
of transitions from $\post{S}$ 
and a sequence $S=R_1\supsetneq R_2\cdots\supsetneq
	R_k\supsetneq Q$
where for $i=1,2,\dots,k$ we have 
$R_i=S\smallsetminus\pre{\{t_1,t_2,\dots,t_{i-1}\}}$,
		$\pre{t_i}\cap R_{i}\neq\emptyset$, and
	$\post{t}_i\subseteq P\smallsetminus R_i$ (which also entails
	that $\pre{t_i}\cap Q=\emptyset$), and, moreover,
$S\smallsetminus\pre{\{t_1,t_2,\dots,t_{k}\}}=Q$.
	By $T\prior$ we denote the set
	$\{t_1,t_2,\dots,t_k\}$ of these fixed (priority) transitions.

	Let us choose $\sigma\in\left((T\smallsetminus
	\post{S})^*\cdot T\prior\right)^*$ and $M$ so that 
	$M_0\gt{\sigma}M$ and there is no $\sigma'\in (T\smallsetminus
	\post{S})^*\cdot T\prior$ enabled at $M$.
	Since the transitions from $T\smallsetminus\post{S}$ do not
	affect the marking on $S$ (recall that
	$\pre{S}\subseteq\post{S}$), the definition of $T\prior$
	guarantees that there must be such $\sigma$ and $M$.
(The transition $t_i\in T\prior$ can occur in $\sigma$ no more than $m_i$
times, where
$m_k=\min\{M_0(p)\mid p\in R_k\smallsetminus Q\}$, and
$m_{i}=\sum_{i<j\leq k}m_j +\min\{M_0(p)\mid p\in R_i\smallsetminus
R_{i+1}\}$ for $i=k{-}1,k{-}2,\dots,1$.) We also note that
$M\restr{Q}=\mathbf{0}$ (since $M_0\restr{Q}=\mathbf{0}$ and we have 
$t\not\in \pre{Q}$ for all $t\in T\prior\cup (T\smallsetminus \post{S})$).

We claim that $\post{S}\subseteq D_{M}$, which shows that our goal
		$\post{S}\cap L_{M_0}=\emptyset$ is satisfied.
For the sake of contradiction, we assume $M\gt{\sigma'}M'$ where 
$\sigma'\in(T\smallsetminus\post{S})^*$ and some $t\in \post{S}$ is enabled
at $M'$. Since $M'\restr{Q}=\mathbf{0}$, we have
$\pre{t}\cap(S\smallsetminus Q)\neq\emptyset$; hence
$\pre{t}\cap\pre{t_i}\neq\emptyset$
for some $t_i\in T\prior$.
The free-choice property thus
guarantees that $t_i$ is enabled at $M'$ as well, which contradicts
our choice of $M$. 
\end{proof}

\bibliographystyle{splncs04}
\bibliography{bibliography}

\end{document}